\documentclass[aps,twocolumn,pra,floatfix,tightenlines,superscriptaddress,longbibliography]{revtex4-1}
\usepackage{graphicx}
\usepackage{epstopdf}
\usepackage[english]{babel}
\usepackage{amsmath}
\usepackage{amssymb}
\usepackage{mathrsfs}
\usepackage{appendix}
\usepackage{bm}
\usepackage{float}
\usepackage{color}
\usepackage{longtable}
\usepackage{url}
\usepackage{enumerate}
\usepackage[usenames,dvipsnames]{xcolor}
\usepackage[colorlinks=true,linkcolor=Blue,urlcolor=Blue,citecolor=Blue]{hyperref}
\usepackage{amsthm}
\newtheorem{theorem}{Theorem}
\newtheorem{lemma}{Lemma}
\newtheorem{definition}{Definition}

\graphicspath{{./Figures}}

\begin{document}

\title{Exponentially improved efficient machine learning for quantum many-body states with provable guarantees}

\author{Yanming Che}
\email{yanmingche01@gmail.com}
\affiliation{Department of Physics, University of Michigan, Ann Arbor, Michigan
48109-1040, USA}
\affiliation{Theoretical Quantum Physics Laboratory, Cluster for Pioneering Research, RIKEN, Wako-shi, Saitama 351-0198, Japan}

\author{Clemens Gneiting}
\email{clemens.gneiting@riken.jp}
\affiliation{Theoretical Quantum Physics Laboratory, Cluster for Pioneering Research, RIKEN, Wako-shi, Saitama 351-0198, Japan}
\affiliation{Center for Quantum Computing, RIKEN, Wako-shi, Saitama 351-0198, Japan}

\author{Franco Nori}
\email{fnori@riken.jp}
\affiliation{Theoretical Quantum Physics Laboratory, Cluster for Pioneering Research, RIKEN, Wako-shi, Saitama 351-0198, Japan}
\affiliation{Center for Quantum Computing, RIKEN, Wako-shi, Saitama 351-0198, Japan}
\affiliation{Department of Physics, University of Michigan, Ann Arbor, Michigan
48109-1040, USA}

\date{\today}

\begin{abstract}
Solving the ground state and the ground-state properties of quantum many-body systems is generically a hard task for classical algorithms. For a family of Hamiltonians defined on an $m$-dimensional space of physical parameters, the ground state and its properties at an arbitrary parameter configuration can be predicted via a machine learning protocol up to a prescribed prediction error $\varepsilon$, provided that a sample set (of size $N$) of the states can be efficiently prepared and measured. In a recent work [Huang \emph{et al.}, Science \textbf{377}, eabk3333 (2022)], a rigorous guarantee for such a generalization was proved. Unfortunately, an exponential scaling for the provable sample complexity, $N = m^{ {\cal{O}} \left(\frac{1}{\varepsilon} \right) }$, was found to be universal for generic gapped Hamiltonians. This result applies to the situation where the dimension of the parameter space is large while the scaling with the accuracy is not an urgent factor, not entering the realm of more precise learning and prediction.
In this work, we consider an alternative relevant scenario, where the effective dimension $m$ is a finite, not necessarily large constant while the scaling with the prediction error becomes the central concern. By jointly preserving the fundamental properties of density matrices in the learning protocol and utilizing the continuity of quantum states in the parameter range of interest, we rigorously obtain a polynomial sample complexity for predicting quantum many-body states and their properties, with respect to the \emph{uniform} prediction error $\varepsilon$ and the number of qubits $n$, with $N = \mathrm{poly} \left(\varepsilon^{-1}, \ n, \ \log \frac{1}{\delta}\right)$, where $\mathrm{poly}$ denotes a polynomial function, and ($1-\delta$) is the probability of success. Moreover, if restricted to learning local quantum-state properties, the number of samples can be further reduced to $N = \mathrm{poly} \left(\varepsilon^{-1}, \ \log \frac{n}{\delta}\right)$. Numerical demonstrations confirm our findings, and an alternative approach utilizing statistical learning theory with reproducing kernel Hilbert space achieves consistent results. The mere continuity assumption indicates that our results are not restricted to gapped Hamiltonian systems and properties within the same phase.
\end{abstract}

\maketitle

\section{Introduction}
Predicting quantum many-body ground states and their properties lies at the heart of a multitude of branches of modern quantum science~\cite{CarleoTroyerScience2017,AlbashRMPAQC,GeorgescuRMPQuSim,CaoChemRev2019}, ranging from condensed-matter physics, quantum computation and simulation, to quantum chemistry. However, finding the exact ground states of quantum many-body Hamiltonians is generically beyond the reach of efficient classical algorithms, under a widely believed conjecture in computational complexity theory that the polynomial hierarchy does not collapse~\cite{AroraCCbook2009}. Such complexity arises from the rich structure of the quantum many-body state space as well as the principles of quantum mechanics. Recently, it was proved~\cite{AbrahamsenSubExp2020,HuangScience2022Manybody} that, for a particular family of gapped Hamiltonians, evaluating ground-state properties such as the expectation values of local operators, can be as hard as solving $\mathrm{NP}$-hard problems in classical settings. 

On the other hand, machine learning as an effective approximation scheme whose solution can be efficiently obtained via proper optimization processes, has found a wide spectrum of applications in solving complex quantum-physics problems~\cite{Schuld2014,Biamonte2017QML,Dunjko2018,CarleoRMPML,DasSarma2019,MehtaPhysRep2019,RinaldiPRXQ2022,ZengPRL2023,DanielNano2024,ShiPRApp2024}, including approximating and predicting quantum states~\cite{CarleoTroyerScience2017,GaoNC2017,EquivPRB2017,DengPRB2017,TorlaiNatPhys2018,ChooPRL2018,GaoPRL2018Exp,NoeScience2019,YoshiokaPRB2019,HartmannPRL2019,NagyPRL2019,VicentiniPRL2019,LuchnikovEntropy2019,InferringDM2019,KwonPRB2019,SharirPRL2020,GlielmoPRX2020,BorinPRB2020,ParkPRR2020,Ohtsuki2020,ZenPRE2020,Barr2020,MelkaniPRA2020,Palmieri_2020,LohaniMLST2020,NeugebauerPRA2020,ShahnawazPRL2020,ShahnawazPRR2021,YoshiokaCommunPhys2021,YaoPRX2021,NomuraPRL2021,TranPRL2021_LTQT,HuangScience2022Manybody,ZhengPRR2022,ChePRB2022,NatPhysRev2023_Complexity}, as well as identifying phases and phase transitions~\cite{LeiWangPRB2016,Ohtsuki2016,WetzelPRE2017Unsupervised,Broecker2017,HuPRE2017Discovering,PontePRB2017Kernel,CarrasquillaNatPhys2017,NieuwenburgNatPhys2017,ZhangPRL2017Quantum,PhysRevX.7.031038,KelvinPRB2018,BeachPRB2018Machine,XinWanPRB2018,YeHuaPRL2018,HuiZhaiPRL2018,HuiZhaiPRB2018,HuembeliPRB2018,YoshiokaPRB2018,Pilozzi2018,VenderleyPRL2018,ZhangPRE2019Machine,HuembeliPRB2019,PhysRevB.99.121104,RemNatPhys2019,PhysRevB.99.104410,Giannetti2019,CanabarroPRB2019,FrankPRE2019,DurrPRB2019,LianPRL2019,Rodriguez-NievaNatPhys2019,ChePRB2020,ScheurerPRL2020,LongPRL2020,LidiakPRL2020,BalabanovPRR2020,Greplova2020,Berezutskii2020,DengPRL2021,Mendes-SantosPRX2021,Mendes-SantosPRXQ2021,LongPRL2023}. 
This raises the hope that machine learning may also help, in some cases, with \emph{efficiently} predicting quantum many-body ground states and their properties; for instance, when the Hamiltonian depends on parameters and a training set of ground states or their properties can be efficiently prepared for some specific parameter choices, in which case learning algorithms can help to generalize to unknown parameters with a lower cost compared to that required in experiments, and to make predictions for the entire parameter space, up to a small prediction error $\varepsilon$. 

Such a situation has recently been addressed in~\cite{HuangScience2022Manybody}, where a rigorous guarantee for learning and predicting properties of quantum many-body ground states was presented, provided that a training set of the ground states is efficiently available via classical shadow tomography~\cite{HuangNatPhys2020,HuangPRL2021}. But unfortunately, a provably exponential scaling of the sample complexity $N$ with the prediction error $\varepsilon$, $N = m^{ {\cal{O}} \left(\varepsilon^{-1} \right) }$, was found to be universal for generic gapped Hamiltonians, where $m$ is the dimension of the parameter space. This situation is therefore mostly useful when the dimension $m$ is large while the error tolerance is high, and precise learning and prediction are not in focus.

An alternative relevant practical scenario may be that the scaling of $N$ with the accuracy is the main concern while the number of parameters $m$ is a finite constant. This includes, for instance, situations where one is only concerned about spatially local properties of the system, or about continuous macroscopic quantities in the critical region of phase transtitions, where only a small number of terms in the Hamiltonian is relevant.

In this work, we are motivated by the framework of the Probably Approximately Correct (PAC) learnable~\cite{Valiant1984,FoundationsMLBook}, which is a fundamental theory of computational learning, and requires the minimal number of samples for the learning task to scale (at most) polynomially in the inverse of the prediction error ($\varepsilon^{-1}$), as well as in the inverse of the failure probability ($\delta^{-1}$). We explore theoretical guarantees for polynomially efficient learning that significantly reduces the provable sample complexity scaling exponentially with $\varepsilon^{-1}$ in~\cite{HuangScience2022Manybody}, and find a decrease of many orders of magnitude in it. The polynomial sample complexity is numerically demonstrated for a quantum $XY$ model, which highlights that our results apply to more general quantum systems beyond gapped Hamiltonians~\cite{BachmannCMP2011} and properties of a single phase.

While using the framework of kernel-based learning, we emphasize \emph{physical constraints} on the learning protocol to avoid possible unphysical predictions. These constraints lead us to introduce the concept of positive good kernels (PGKs)~\cite{FourierBookStein}, which will be elaborated below. 

Moreover, kernel-based models for learning and predicting quantum states have the following advantages: First, kernel methods have been widely used in supervised and unsupervised learning tasks such as nonlinear regression and classification~\cite{SVMBook2008,FoundationsMLBook}, and also dimensionality reduction~\cite{coifman2005geometric,nadler2006diffusion,Rodriguez-NievaNatPhys2019,ChePRB2020,ScheurerPRL2020}. The kernel itself contains rich information such that there is no need to calculate the feature maps in the high-dimensional feature space~\cite{YuTing2017,BartkiewiczSciRep2020}, which drastically reduces the computational cost. Second, even supervised learning with neural networks, the currently most popular architecture, is equivalent to the neural tangent kernel (NTK) regression in the limit of infinite width of the network~\cite{NTK2018}. Third, each positive kernel is uniquely associated with a reproducing kernel Hilbert space (RKHS)~\cite{SVMBook2008,FoundationsMLBook,MinhJMLR2016,YuTing2017,arxiv2101.11020}, which can be a universal hypothesis space for the learning problem (in the case that the kernel is a universal kernel). With these properties, kernel-based methods not only are of practical value, but also can provide a unified framework of theoretical analysis for various provable guarantees.

\section{Setting and model}
We start from a continuously parametrized family of $n$-qubit quantum states $\rho(x)$ with $x \in {\cal{X}}$, where ${\cal{X}} = [-\frac{L}{2}, \frac{L}{2}]^m \subset \mathbb{R}^m$ is an $m$-dimensional space of physical parameters, with $L$ the length of the intervals in each dimension. Without loss of generality, here we restrict to qubit models. For fermionic models, they can be encoded as qubits through, for example, the Jordan-Wigner~\cite{JordanZphysik1928} or the Bravyi-Kitaev~\cite{BravyiAnnPhys2002} transformation. The states $\rho(x)$ may be, but are not restricted to, the ground states of a parametrized Hamiltonian $H(x)$. Additionally, we assume that the parameter space is equipped with a probability measure $\mu$, according to which a sample $x_i$ is drawn from $\cal{X}$, and that a training set of quantum states, ${\cal{S}} = \{x_i, \rho(x_i) \}_{i=1}^N$, can be efficiently prepared classically or quantumly. 

While the tomography for accessing $\rho(x_i)$ experimentally is in general expensive for multi-qubit systems, various strategies have been proposed to reduce the cost of quantum state tomography, e.g., via compressed sensing~\cite{GrossPRL2010,FlammiaNJP2012}, or with neural-network quantum state tomography~\cite{TorlaiNatPhys2018,MelkaniPRA2020,Palmieri_2020,LohaniMLST2020,NeugebauerPRA2020,ShahnawazPRL2020,ShahnawazPRR2021}. In particular, in classical shadow tomography~\cite{HuangNatPhys2020,HuangPRL2021, HuangScience2022Manybody}, the sample state is approximated via randomized Pauli measurements $\rho(x_i) \approx \sigma_T(x_i)$, with $\sigma_T(x_i)$ the classical shadow state over $T$ copies of measurements at the parameter $x_i$. Moreover, with locality assumptions, the approximation via shadow states has been proved~\cite{HuangNatPhys2020,HuangScience2022Manybody} to be very efficient, requiring only $T = {\cal{O}} \left(\log n / \epsilon^2 \right)$ copies of measurements, where $\epsilon$ is the precision of the tomography in trace norm. For more details about the classical shadow tomography, see~\cite{HuangNatPhys2020,HuangScience2022Manybody} 

Without loss of generality, the target functions to be learned, such as the real and imaginary parts of the density-matrix entries, or the order parameters of quantum phases, are assumed to be \emph{continuous} real functions $f: {\cal{X}} \to \mathbb{R}$~\cite{ContinuityNote}. Furthermore, periodic boundary conditions can be assumed for $\cal{X}$, as in~\cite{HuangScience2022Manybody}, such that the function $f$ is defined on a circle of length $L$ in each dimension.
Finally, the predicted density matrix at an arbitrary $x \in {\cal{X}}$ is given by the kernel generalization from the training set ${\cal{S}}$, 
\begin{eqnarray}
\label{eq:Kernel_DM}
\sigma_N (x) = \frac{1}{N} \sum_{i=1}^N K_{\Lambda} (x-x_i) \rho(x_i),
\end{eqnarray}
where we assume a kernel with translational symmetry, i.e., 
\begin{eqnarray}
K_{\Lambda} (x, x_i) = K_{\Lambda} (x-x_i),
\end{eqnarray}
with positive integer index $\Lambda$.
 
Note that, the estimated density matrix $\sigma_N (x)$ in (\ref{eq:Kernel_DM}) should preserve the fundamental properties of density operators for all $x \in {\cal{X}}$, i.e., it should be Hermitian, positive semi-definite, and have unit trace. The first two conditions require $K_{\Lambda}(x) \ge 0$ to be real, while the last one, due to approximation and statistical errors, can be approximately satisfied in a practical learning process, i.e., 
\begin{eqnarray}
\mathrm{Tr} \sigma_N(x) = \frac{1}{N} \sum_{i=1}^N K_{\Lambda} (x-x_i) \approx 1.
\end{eqnarray}

\section{Positive good kernels (PGKs)}
In order to facilitate an efficient and accurate learning of the quantum state (\ref{eq:Kernel_DM}) satisfying the physical constraints for density matrices, we characterize the following conditions for PGKs (e.g., see~\cite{FourierBookStein,PGKNote}): 
\newline
(I) Positivity and boundness: 
\begin{eqnarray}
0 \le K_{\Lambda}(x) \le {\cal{O}}\left(\Lambda^{\tau} \right) \ \left(\forall x \in {\cal{X}} \right),
\end{eqnarray}
with $\tau$ some positive integer; 
\newline
(II) Normalization: 
\begin{eqnarray}
\int_{x \in {\cal{X}}} K_{\Lambda}(x)  \mathrm{d}\mu(x) = 1;
\end{eqnarray}
\newline
(III) $\eta$-convergence: For all $0 < \eta \le L$, 
\begin{eqnarray}
\int_{x \in {\cal{X}}, \ \| x\|_2 \ge \eta} | K_{\Lambda}(x) | \mathrm{d}\mu(x) \le {\cal{O}}\left( \Lambda^{-1}\right).
\end{eqnarray}

The normalization (II) of the kernel guarantees that $\mathrm{Tr} \sigma_N(x) \approx 1$ with high probability with respect to ${\cal{S}}$. The conditions (I, II, III) allow us to predict the continuous quantum state within a prescribed uniform precision and in an efficient manner (see more details and examples of PGKs in Appendix~\ref{sec:ProofThms}). If the uniform distribution on $\cal{X}$ is used, the probability measure will be given by $\mathrm{d}\mu(x) = \frac{1}{L^m} \mathrm{d} x$. 
As an example, we present a rectangular Fej\'er kernel, which is a PGK (see Appendix~\ref{sec:ProofThms} for proof details) and which is given by~\cite{FourierBookStein,Pfister2019BoundingMT} 
\begin{eqnarray}
\label{eq:Fejer_Closed}
F_{\Lambda} (x) = \frac{1}{\Lambda^m} \prod_{i = 1}^m \frac{\sin^2 \left( \frac{\Lambda \pi}{L} x_i\right)}{\sin^2 \left( \frac{\pi x_i}{L} \right)}, 
\end{eqnarray}
where we have $0 \le F_{\Lambda} (x) \le \Lambda^m$.

In~\cite{HuangScience2022Manybody}, an $\ell_2$-Dirichlet kernel, $D_{\Lambda} (x)$, has been shown to perform well in approximating smooth functions; in particular the average of local observables with respect to the ground state, with a bounded derivative of the first-order and through a truncated Fourier series. However, we find that, in the task of kernel learning of the density matrix, the Dirichlet kernel may fail to be a PGK, in the sense that it violates the positivity and the $L_1$-boundness~\cite{FourierBookStein}. For example, taking the dimension $m=1$, the entries of $D_{\Lambda} (x)$ can be negative for some values of $x$, and also, $\frac{1}{2\pi}\int_{-\pi}^{\pi} | D_{\Lambda} (x) | \mathrm{d}x \ge c \log \Lambda $, with some constant $c>0$. 

\section{Results} 
We are now prepared to present our main findings exposed by the following two theorems. 

\begin{theorem}[Efficient learning of quantum-state representations with positive good kernels (PGKs)] 
\label{thm:LearningDM}
\textbf{Given:} A family of $n$-qubit unknown quantum states $\rho(x)$ with continuously~\cite{ContinuityNote} parametrized density-matrix entries defined on the parameter space ${\cal{X}} \subset \mathbb{R}^m$, where $m$ a constant; and a training set ${\cal{S}} = \{x_i, \rho(x_i) \}_{i=1}^N$ with $x_i$ drawn from $\cal{X}$ according to a probability density associated with $\mu$.
\medskip
\newline
\textbf{Goal:} Output an estimator $\sigma_N(x)$ for the density matrix which preserves the positivity and the unit trace through $(\ref{eq:Kernel_DM})$ via the PGKs, with an as small as possible number of samples $N$, and which is close to the true density matrix $\rho(x)$ for an arbitrary $x \in {\cal{X}}$, in terms of the uniform prediction error 
\begin{eqnarray}
\underset{x \in {\cal{X}}}{\sup} \left\| \sigma_N(x) - \rho(x) \right\|_{\ell_{\infty}} \le \varepsilon 
\end{eqnarray}
and
\begin{eqnarray}
\left| \mathrm{Tr} \sigma_N(x) - 1 \right| \le \varepsilon 
\end{eqnarray}
at the same time, with probability at least $1-\delta$ $(0 < \delta < 1)$.
\medskip
\newline
\textbf{Result (Sample complexity):} The above goal can be achieved with only 
\begin{eqnarray}
\label{eq:Nsample_DM}
N = \mathrm{poly} \left( \varepsilon^{-1}, \ n, \ \log \frac{1}{\delta} \right). 
\end{eqnarray}
Note that the matrix $\ell_{\infty}$-norm here (different from the operator norm) is defined as $\left\| \rho\right\|_{\ell_{\infty}} = \underset{i, j}{\max} \left| \rho_{ij} \right|$, with $\rho_{ij}$ the entries of the matrix $\rho$, and $\mathrm{poly}$ denotes a polynomial function. 
\end{theorem}

While the computational time can generically depend exponentially on the number of qubits, because the number of entries of the density matrix is of ${\cal{O}} \left( 4^n\right)$, this may be reduced to $\mathrm{poly} (n)$ for sparse states. We emphasize that $\varepsilon$ represents the uniform prediction error for the whole parameter space. Theorem~\ref{thm:LearningDM} guarantees that efficient learning and prediction are performed physically by (approximately) preserving the fundamental properties of the density matrix. For more discussions about the continuity of the density matrix, see Appendix~\ref{sec:SmoothParametrization}.

Motivated by the importance as well as the hardness of computing the ground-state properties for generic quantum many-body Hamiltonians, we next relax to the target of learning a state property, 
\begin{eqnarray}
f_O (x) = \mathrm{Tr} \left[ O \rho(x)\right] = \sum_{i=1}^{M} f_i(x),
\end{eqnarray}
a continuous function on ${\cal{X}}$ given by the expectation value of a $q$-local operator $O = \sum_{i=1}^{M} O_i$, with respect to the quantum many-body state $\rho(x)$. The $i$th-partite local state property reads 
\begin{eqnarray}
f_i(x) = \mathrm{Tr} \left[ O_i \rho(x)\right],
\end{eqnarray}
with $i \in [M] = \{1, 2, \cdots, M \}$, where $M$ is a positive integer scaling polynomially in $n$, and the local function is bounded by $\left| f_i (x) \right| \le {\cal{O}} (1)$.

\begin{figure*}[t]
\centerline{\includegraphics[height=2.5in,width=6in,clip]{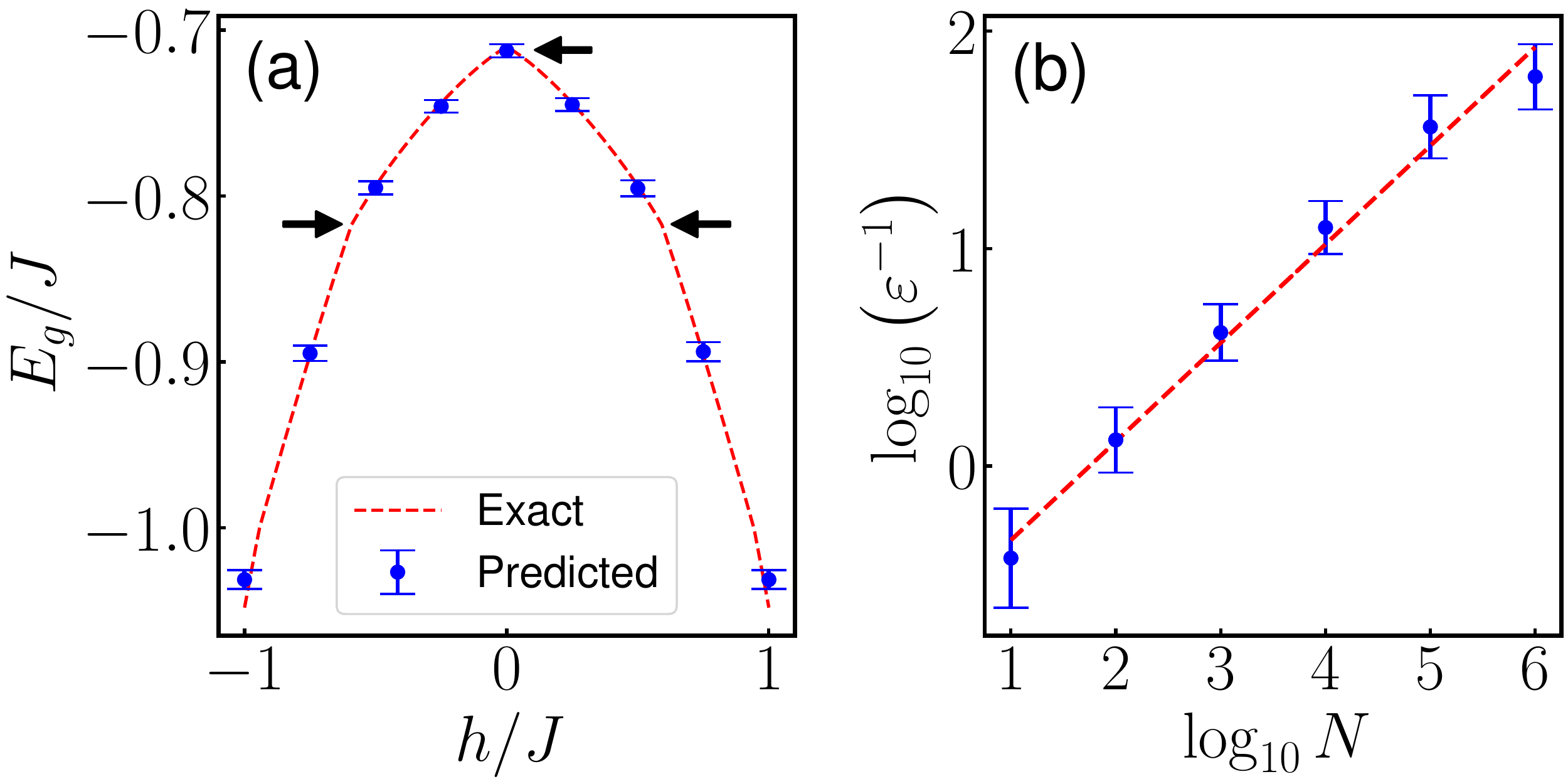}}
\caption{(a) Comparison between the exact (red dashed curve) ground-state energies per qubit and the positive good kernel (PGK) predicted (blue dots) ones, for a finite-size one-dimensional quantum $XY$ model with sample size $N = 10^6$. The black arrows indicate the locations where the energy curve is \emph{continuous but not smooth} due to the vacua competition~\cite{PasqualePRA2009,OkuyamaPRE2015}. Note that there are two additional but not obvious points in the vicinity of $h/J = \pm 1$ where the energy curve is also continuous but not smooth due to the vacua competition, which are not indicated in this figure. 
(b) Double-logarithmic plot of $\varepsilon^{-1}$ vs $N$ (blue dots), and a linear regression (red dashed line) with a slope $\approx 0.45$ and an R-squared score $\approx 0.99$, confirming a polynomial scaling of the sample complexity. We use the rectangular Fej\'er kernel with $\Lambda = 50$, $n=5$ qubits, $\gamma = 1/3$, and samples from a uniform distribution in both (a) and (b). The error $\varepsilon$ in (b) is obtained as the maximal error between the predicted and true energies. The mean values and error bars in both panels are calculated with $30$ independent runs.
}
\label{fig:XY_model}
\end{figure*}

\begin{theorem}[Efficient learning of quantum-state properties with positive good kernels (PGKs)]
\label{thm:LearningGSP}
\textbf{Given:}
A training set of size $N$ efficiently prepared as ${\cal{S}} = \{x_i, f_O (x_i) \}_{i=1}^N$.
\medskip
\newline
\textbf{Goal:} Predict the quantum-state property 
\begin{eqnarray}
\hat{f}_O (x) =  \mathrm{Tr} \left[ O \sigma_N (x)\right] = \sum_{i=1}^{M} \hat{f}_i(x)
\end{eqnarray}
via the positive $\sigma_N (x)$ with unit trace in $(\ref{eq:Kernel_DM})$, with an as small as possible number of samples $N$, and such that it is close to the true values of $f_O (x)$ in terms of the uniform prediction error 
\begin{eqnarray}
\underset{x \in {\cal{X}}}{\sup} \left| \hat{f}_O (x) - f_O (x)\right| \le \varepsilon,
\end{eqnarray}
with probability at least $1-\delta$ $(0 < \delta < 1)$.
\medskip 
\newline
\textbf{Result (Sample complexity):} The above goal can be achieved with only 
\begin{eqnarray}
\label{eq:Nsample_fo}
N = \mathrm{poly} \left(\varepsilon^{-1}, \ n, \ \log \frac{1}{\delta}\right).
\end{eqnarray}
Moreover, if locality is assumed, where the learning only requires $\underset{x \in {\cal{X}}}{\sup} \left| \hat{f}_i (x) - f_i (x)\right| \le \varepsilon$ $(\forall i \in [M])$ with probability at least $1-\delta$, then $N$ can be further reduced to 
\begin{eqnarray}
\label{eq:Nsample_fi}
N = \mathrm{poly} \left(\varepsilon^{-1}, \ \log \frac{n}{\delta}\right).
\end{eqnarray}
Now the computational time scales only (at most) polynomially in the number of qubits $n$.
\end{theorem}

For a summary of the main ideas and the detailed proofs of the two theorems, see Appendix~\ref{sec:ProofThms}.

\begin{figure*}[t]
\centerline{\includegraphics[height=2.5in,width=6in,clip]{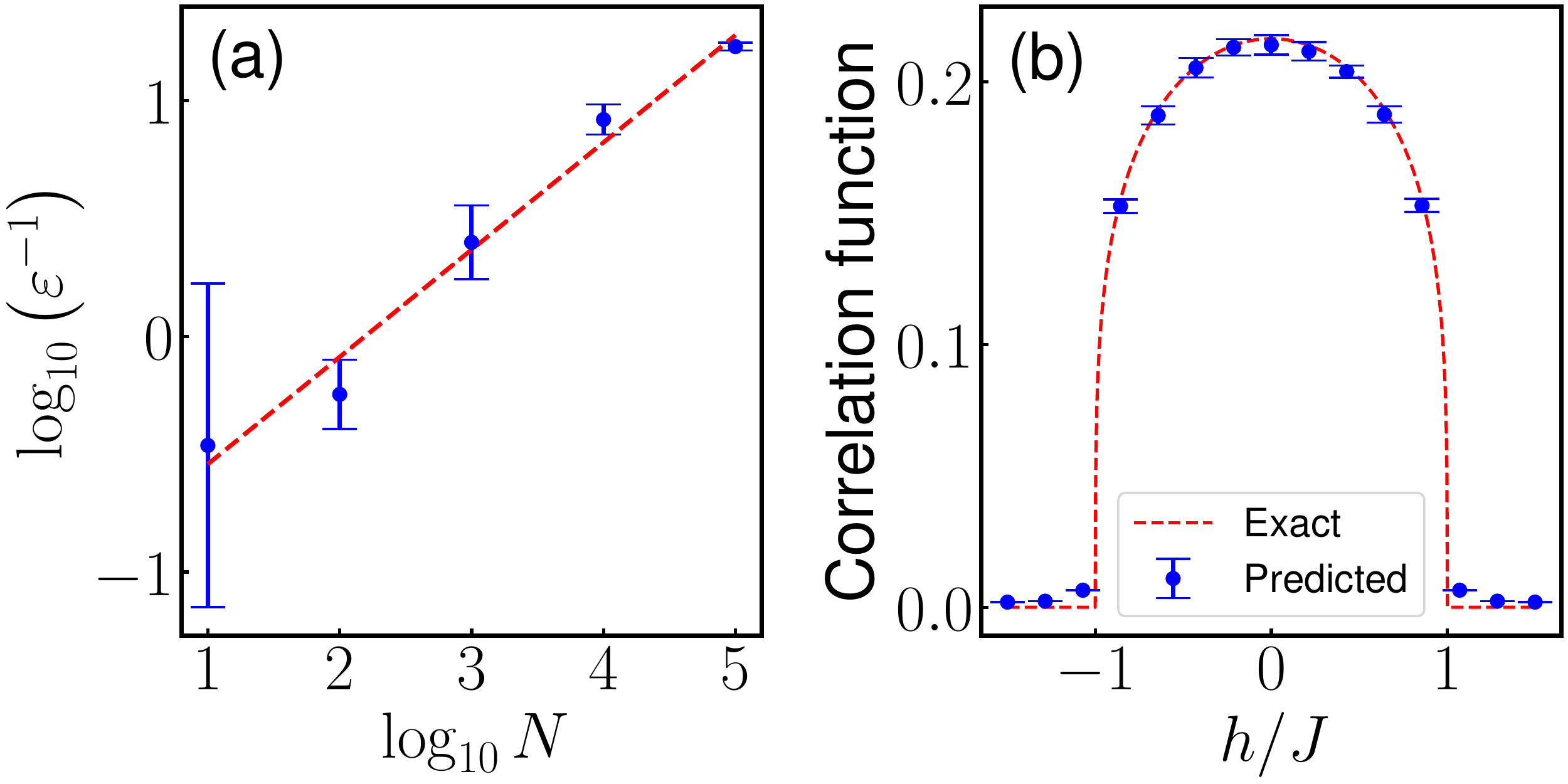}}
\caption{Predicting the ground-state long-range order (spin-spin correlation function) $\lim_{r \rightarrow \infty}(-1)^r \langle S_0^x S_r^x\rangle$ of the quantum $XY$ model in the thermodynamic limit.
(a) Double-logarithmic plot of the positive good kernel (PGK) prediction accuracy $\varepsilon^{-1}$ vs $N$ (blue dots) within the two-dimensional parameter space $-3/2 \le h/J \le 3/2$ and $0 \le \gamma \le 1$. A linear regression (red dashed line) with a slope $\approx 0.45$ and an R-squared score $\approx 0.98$ confirms a polynomial scaling of the sample complexity. 
(b) Comparison between the exact (red dashed curve) long-range spin-spin correlation function and the PGK predicted (blue dots) one with $N = 10^5$, for a quantum $XY$ chain with one-dimensional parameter space $-3/2 \le h/J \le 3/2$ and a fixed value of $\gamma = 1/3$. We use uniform distributions, a Fej\'er kernel with $\Lambda = 50$, and $30$ independent runs as in Fig.~\ref{fig:XY_model}. 
}
\label{fig:XY_CorrelationFunction}
\end{figure*}

\section{Comparison with existing results}
We compare the provable polynomial sample complexity here with that scaling exponentially with respect to the prediction accuracy while being restricted to learning the ground-state properties $f_O(x)$ in~\cite{HuangScience2022Manybody}. With a unified error metric and the same notation, the provable sample complexity in~\cite{HuangScience2022Manybody} is of the order 
\begin{eqnarray}
N_0 =\frac{B^2}{\varepsilon^2} (2m + 1)^{\frac{1}{\varepsilon^2}}.
\end{eqnarray}
For instance, taking $m = 2$, the number of qubits $n=100$, a $q$-local operator $O$ such that $M = {\cal{O}}(n^q)$ (taking $q = 1$), and $\varepsilon = 0.1$, we then have $N/N_0 \approx 10^{-48}$ with a Fej\'er kernel, and $N/N_0 \approx 10^{-61}$ with a Gaussian kernel, respectively, where $N$ denotes the number of samples in (\ref{eq:Nsample_fo}) (see more details in Appendix~\ref{sec:ComparisonExistingResults}). This demonstrates that the rigorous guarantees provided here are more powerful when $m$ is not large and the prediction precision is a central concern. 

\section{Numerical demonstration} 
While our rigorous results target experimentally obtained training sets, here we demonstrate the above findings regarding learning the ground-state properties with a one-dimensional ($1$D) quantum $XY$ model~\cite{SachdevQPTBook,PasqualePRA2009,OkuyamaPRE2015}. The Hamiltonian is given by 
\begin{eqnarray}
H = -J \sum_{i=1}^{n} \frac{1 + \gamma}{2} \sigma_i^x \sigma_{i+1}^x + \frac{1 - \gamma}{2} \sigma_i^y \sigma_{i+1}^y + \frac{h}{J} \sigma_i^z, \ \ \
\end{eqnarray}
where $\sigma^{x, y, z}_i$ are the Pauli matrices of the $i$-th qubit, $J > 0$ and $h$ parametrize the nearest-neighbour interaction and the transverse field, respectively, and $0 \le \gamma \le 1$ is the anisotropic parameter. We assume periodic boundary conditions for the qubits. This model is important for studying quantum many-body physics and is exactly solvable in the dual fermion picture~\cite{JordanZphysik1928,BravyiAnnPhys2002}. The latter allows us to test our theoretical results, for example, when $f_O(x)$ is the ground-state energy or the correlation function. 

An interesting feature of the $1$D quantum $XY$ model for finite sizes is that there are two fermion number parity sectors in the Jordan-Wigner fermion picture~\cite{PasqualePRA2009,OkuyamaPRE2015}, corresponding to periodic and anti-periodic boundary conditions of fermions, respectively. The ground state of the model is the winner of the competition between the vacuum states of the two parity sectors. Therefore, there can be gapless points where the ground-state energy curve is \emph{continuous but not smooth} within the parameter range $| h/J | \le 1$ (e.g., see the locations indicated by black arrows in Fig.~\ref{fig:XY_model} (a)). Note that this situation is not covered by the results in~\cite{HuangScience2022Manybody}.

Figure~\ref{fig:XY_model} (a) shows a comparison between the PGK predicted (blue dots) and exact (red dashed curve) energies $E_g / J$; The double-logarithmic plot in Fig.~\ref{fig:XY_model} (b) depicts the scaling of $\varepsilon^{-1}$ with $N$ (blue dots), where a good linear behavior (red dashed line) indicates a polynomial sample complexity. Note that, with a sufficiently large fixed value of $\Lambda$, the generalization error is dominant, then $N \approx {\cal{O}}\left( \varepsilon^{-2}\right)$ is observed, as expected from our results (see Appendix~\ref{sec:BoundGeneralizationError}). We use $n = 5$ qubits and $\gamma = 1/3$ to sharpen the non-smooth behavior of the energy curve. The excellent agreement between exact and predicted ground-state energies confirms that our provably rigorous guarantee is not limited to smoothly parametrized or gapped Hamiltonian systems as in~\cite{HuangScience2022Manybody}, but can apply to more general quantum systems.

As a second example, we provide results for predicting the long-range spin-spin correlation functions of the quantum $XY$ model with one- and two-dimensional parameter spaces, respectively. The target is to learn the ground-state two-point function~\cite{BarouchPRA1971} 
\begin{eqnarray}
\lim_{r \rightarrow \infty}(-1)^r \langle S_0^x S_r^x\rangle, 
\end{eqnarray}
where the spin operators are given by $S_i^{x, y, z} = \sigma_i^{x, y, z}/2$ (set $\hbar = 1$), and the training set is sampled from the exact solution in the thermodynamic limit~\cite{BarouchPRA1971}.

In Fig.~\ref{fig:XY_CorrelationFunction} (a) we present a double-logarithmic plot of the inverse uniform error $\varepsilon^{-1}$ versus the number of samples $N$ (blue dots), obtained within the two-dimensional parameter space $(h, \gamma)$. The linear scaling (red dashed line) indicates a polynomial sample complexity. In Fig.~\ref{fig:XY_CorrelationFunction} (b), the comparison between the exact (red dashed curve) and the predicted (blue dots) correlation functions is plotted, with the parameter $h$ varying and a fixed value of $\gamma$. More details can be found in the caption of Fig.~\ref{fig:XY_CorrelationFunction}. We note that, in both Figs.~\ref{fig:XY_CorrelationFunction} (a, b), the range of parameters includes values where the correlation function has a divergent or ill-defined gradient (i.e., it is not smooth), e.g., at $|h/J| = \pm 1$ in Fig.~\ref{fig:XY_CorrelationFunction} (b). This is another strong evidence that our method is applicable to more general cases beyond gapped Hamiltonians as well as properties of the system within the same phase.

\section{Learning quantum-state properties in reproducing kernel Hilbert space (RKHS)}
Now we further elaborate the kernel model, and exploit the theory of generalization in machine learning~\cite{FoundationsMLBook}, to explicate the result in Theorem~\ref{thm:LearningGSP}. 
We recall that in~\cite{SVMBook2008,arxiv2101.11020} each symmetric and positive definite kernel $K_{\Lambda} (x, x^{\prime})$ defined on ${\cal{X}} \times {\cal{X}}$ is associated to a unique RKHS ${\cal{H}}_{K_{\Lambda}}$ (see a brief review in Appendix~\ref{sec:ReviewRKHS}), with the reproducing property 
\begin{eqnarray}
f(x) = \left\langle f, K_{\Lambda} (x, )\right \rangle_{{\cal{H}}_{K_{\Lambda}}}, 
\end{eqnarray}
for all $f \in {\cal{H}}_{K_{\Lambda}}$. It is straightforward to verify that the estimator for the state property $\hat{f}_O(x)$ given by (\ref{eq:Kernel_DM}) lies in the RKHS ${\cal{H}}_{K_{\Lambda}}$. For universal kernels, the corresponding RKHS is a suitable ansatz space for learning the quantum-state properties. 

Provided a training set ${\cal{S}} = \{x_i, f_O(x_i) \}_{i=1}^N$, the optimal function $\hat{f} \in {\cal{H}}_{K_{\Lambda}}$ can be found by minimizing the empirical error
\begin{eqnarray}
{\cal{E}}_t = \frac{1}{N} \sum_{i=1}^N \left| \hat{f}_O(x_i) - f_O(x_i)\right|,
\end{eqnarray}
with a bounded functional norm $\| \hat{f}_O \|_{{\cal{H}}_{K_{\Lambda}}} \le \lambda_f$. The representer theorem~\cite{SVMBook2008,FoundationsMLBook} states that the optimal solution admits the form of 
\begin{eqnarray}
\hat{f}_O(x) = \sum_{i=1}^N \alpha_i K_{\Lambda}(x - x_i), 
\end{eqnarray}
where the real coefficients $\{\alpha_i\}$ are dual variables determined by minimizing ${\cal{E}}_t$.

To obtain a theoretical guarantee with the PGK estimator in (\ref{eq:Kernel_DM}), the dual variables can be set to be $\alpha_i \approx f_O(x_i) /N$. Moreover, according to the generalization theory of statistical learning, the expected error in terms of the $L_1$-norm, 
\begin{eqnarray}
{\cal{E}}_p = \underset{x \sim \mu}{\mathbb{E}} \left| \hat{f}_O(x) - f_O(x)\right|, 
\end{eqnarray}
is upper bounded by (with probability at least $1-\delta$) (see more details in Appendix~\ref{sec:GSPinRKHS}), 
\begin{eqnarray}
\label{eq:PredictionErrorRKHS}
{\cal{E}}_p \le {\cal{E}}_t + \frac{8 \lambda_f R}{\sqrt{N}} \sqrt{\frac{\log \frac{2}{\delta}}{2}},
\end{eqnarray}
provided that $K_{\Lambda}(x, x) \le R^2$ for all $x \in {\cal{X}}$.

With a PGK estimator, the empirical error ${\cal{E}}_t$ can be upper bounded at ${\cal{O}}(\varepsilon)$ with large probability, given a polynomial $N$ as in (\ref{eq:Nsample_fo}) (see Appendix~\ref{sec:GSPinRKHS} for details). Moreover, the second term on the right-hand side of~(\ref{eq:PredictionErrorRKHS}) can also be consistently bounded by ${\cal{O}}(\varepsilon)$. We then obtain a consistent but relaxed result with respect to Theorem~\ref{thm:LearningGSP}, i.e., the expected prediction error can be of ${\cal{O}}(\varepsilon)$, \emph{requiring only a polynomially efficient sample complexity}.

\section{Conclusion} 
We have provided a theoretical guarantee for \emph{polynomially} efficient machine learning quantum many-body states and their properties, provided that efficient access to the training set is available (e.g., via classical shadow tomography), and that the dimension of the parameter space is not necessarily a large constant.

The key idea underlying our finding is to jointly utilize the continuity of quantum states in the parameter space and to preserve the fundamental properties of the density operator in the learning protocol. Our provably rigorous results reveal that a polynomial sample complexity in terms of the 
\emph{uniform} prediction error and the qubit number is possible, allowing us to predict the quantum states and their properties efficiently and accurately. When restricted to learning local quantum-state properties, the dependence of the number of samples on the number of qubits can be further reduced to $\mathrm{poly} (\log n)$. 

We emphasize again that our results are not restricted to learning smooth ground states of gapped Hamiltonians as in~\cite{HuangScience2022Manybody}. Because our mere continuity assumption of quantum states is more general, it also applies to gapless Hamiltonian systems or continuously parametrized mixed states and their properties (e.g., quantum states smoothly parametrized by Kraus operators~\cite{NielsenChuang}, and quantum metrological states with finite amounts of quantum Fisher information~\cite{Helstrom,Holevo,MaPhysRep2011,ChePRA2019,LiuJPAReview2020,RinaldiQST2024}; see Appendix~\ref{sec:SmoothParametrization} for more details). An alternative approach from the standard generalization theory of statistical learning, based on the complexity of the RKHS, also achieves consistent results.


\textit{Note added.} After this work was completed, two related preprints~\cite{LewisNC2024,arxiv2301.12946} appeared, which obtained a quasi-polynomial sample complexity, while focusing on learning the average of observables under the assumption of geometric locality, accordingly with more restricted and specific learning models.

\emph{Acknowledgments.} 
C.G. is partially supported by a RIKEN Incentive Research Grant. 
F.N. is supported in part by: Nippon Telegraph and Telephone Corporation (NTT) Research, 
the Japan Science and Technology Agency (JST) [via the Quantum Leap Flagship Program (Q-LEAP), 
and the Moonshot R\&D Grant Number JPMJMS2061], 
the Asian Office of Aerospace Research and Development (AOARD) (via Grant No. FA2386-20-1-4069), 
and the Office of Naval Research (ONR) Global (via Grant No. N62909-23-1-2074).

\clearpage
\appendix
\setcounter{figure}{0}
\renewcommand{\thefigure}{A\arabic{figure}}

\begin{widetext}
\section{List of abbreviations and symbols} 
\label{sec:Symbols}

For the readers' convenience, here in Table~\ref{tb:ListSymbols} we list the most-often-used abbreviations and symbols in this work.

\begin{table*}[b]
\centering
\caption{List of abbreviations and symbols. }
\medskip
\begin{tabular}{c l}
\hline \hline
Abbreviation   &   \hskip 30ex    Full form \\ \hline
PAC-learnable &  \hskip 5ex  Probably Approximately Correct-learnable. \\
PGK &  \hskip 5ex  Positive good kernel. \\
NTK &  \hskip 5ex  Neural tangent kernel. \\
RKHS &  \hskip 5ex  Reproducing kernel Hilbert space. \\
A-G-P &  \hskip 5ex  Approximation-Generalization-Prediction. \\
\hline \hline
Symbol   &   \hskip 30ex    Description \\ \hline
${\cal{X}}$ &  \hskip 5ex  Input parameter space, with ${\cal{X}} \subset \mathbb{R}^m$. \\ 
$m$ &  \hskip 5ex  Dimension of the parameter space. \\ 
$L$ &  \hskip 5ex  Length of ${\cal{X}}$ in each dimension.  \\ 
$\mu$ &  \hskip 5ex  Probability measure.  \\ 
${\cal{S}}$ &  \hskip 5ex  Training set. \\
$N$ &  \hskip 5ex  Number of samples/Size of the training set.  \\ 
$n$ &  \hskip 5ex  Number of qubits in the system.  \\ 
$\varepsilon$ &  \hskip 5ex  Prediction error bound of the machine learning model.  \\ 
$\delta$ &  \hskip 5ex  Probability of failure ($1-\delta$ is the success probability).  \\ 
$\eta$ &  \hskip 5ex  Quantity used in the definition of continuous functions.  \\ 
$\Lambda$ &  \hskip 5ex  Index of the kernel/Scale of cutoff.  \\ 
${\cal{K}}_{\Lambda} (x, x^{\prime})$ &  \hskip 5ex  Kernel defined on ${\cal{X}} \times {\cal{X}}$.  \\
${\cal{K}}_{\Lambda} (x-x^{\prime})$ &  \hskip 5ex  Kernel with translational symmetry.  \\ 
${\cal{M}}$ &  \hskip 5ex  Upper bound of the $L_1$ norm of the kernel.  \\ 
$\sigma_N(x)$ &  \hskip 5ex  Kernel density matrix estimator \ $\forall x \in {\cal{X}}$.  \\ 
$\mathrm{Tr}$ &  \hskip 5ex  Trace.  \\ 
$\mathrm{poly}(\alpha, \gamma, \dots)$ &  \hskip 5ex  Polynomial function of $\alpha$, $\gamma$, etc.  \\ 
$d_{\Lambda}(f)$ &  \hskip 5ex  Distance (error) from the convolution approximation of $f$ with PGKs.  \\ 
$d_N(f)$ &  \hskip 5ex  Distance (error) from the generalization over the training set.  \\ 
$O$ &  \hskip 5ex  $q$-local observable defined as $O = \sum_{i=1}^M O_i$, with $M = {\cal{O}}(n^q)$.  \\ 
$B$ &  \hskip 5ex  Twice the upper bound of a continuous function $f(x)$ defined on ${\cal{X}}$.  \\
$C_L$ &  \hskip 5ex  Lipschitz constant.  \\ 
${\cal{H}}_{K_{\Lambda}}$ &  \hskip 5ex  Reproducing kernel Hilbert space (RKHS) associated with kernel $K_{\Lambda}$.  \\
$\alpha_i$ &  \hskip 5ex  Dual variables for $i = 1, 2, \cdots, N$.  \\ 
$\hat{f}_O(x)$ &  \hskip 5ex  Kernel estimator of the expectation value of $O$ with respect to the quantum state $\rho(x)$.  \\ 
$\lambda_f$ &  \hskip 5ex  Upper bound of the functional norm of $\hat{f}_O(x)$ in the RKHS.  \\ 
${\cal{E}}_p$ &  \hskip 5ex  Expected error in terms of $L_1$ norm.  \\ 
${\cal{E}}_t$ &  \hskip 5ex  Corresponding empirical error/training error with respect to ${\cal{E}}_p$ .  \\ 
${\cal{R}_{{\cal{S}}}}$ &  \hskip 5ex  Empirical Rademacher complexity.  \\ 
$R$ &  \hskip 5ex  Square root of the upper bound of the kernel $K_{\Lambda}$.  \\ 
$D_{\mathbf{n}}(x) $ &  \hskip 5ex  Rectangular Dirichlet kernel with the vector index $\mathbf{n}$.  \\
$F_{\Lambda}(x) $ &  \hskip 5ex  Rectangular Fej\'er kernel with the cutoff scale $\Lambda$.  \\
$K_{h}(x)$ &  \hskip 5ex  Gaussian kernel parametrized by $h$.  \\
\hline \hline
\end{tabular}
\label{tb:ListSymbols}
\end{table*}
\clearpage
\end{widetext}

\section{A unified proof for Theorems 1 and 2 in the main text}
\label{sec:ProofThms}

\subsection{Setting and summary of main ideas}
\label{sec:MainIdeas}

\textit{Summary of proof ideas.}---Because the target functions to be learned in Theorems $1$ and $2$ in the main text, such as the entries of the density matrix or the average of the local operators, are continuous functions defined on the parameter space $\cal{X}$, the proofs for the two theorems can be given in a unified framework. 

\begin{figure}[t]
\centerline{\includegraphics[height=3.5in,width=3.5in,clip]{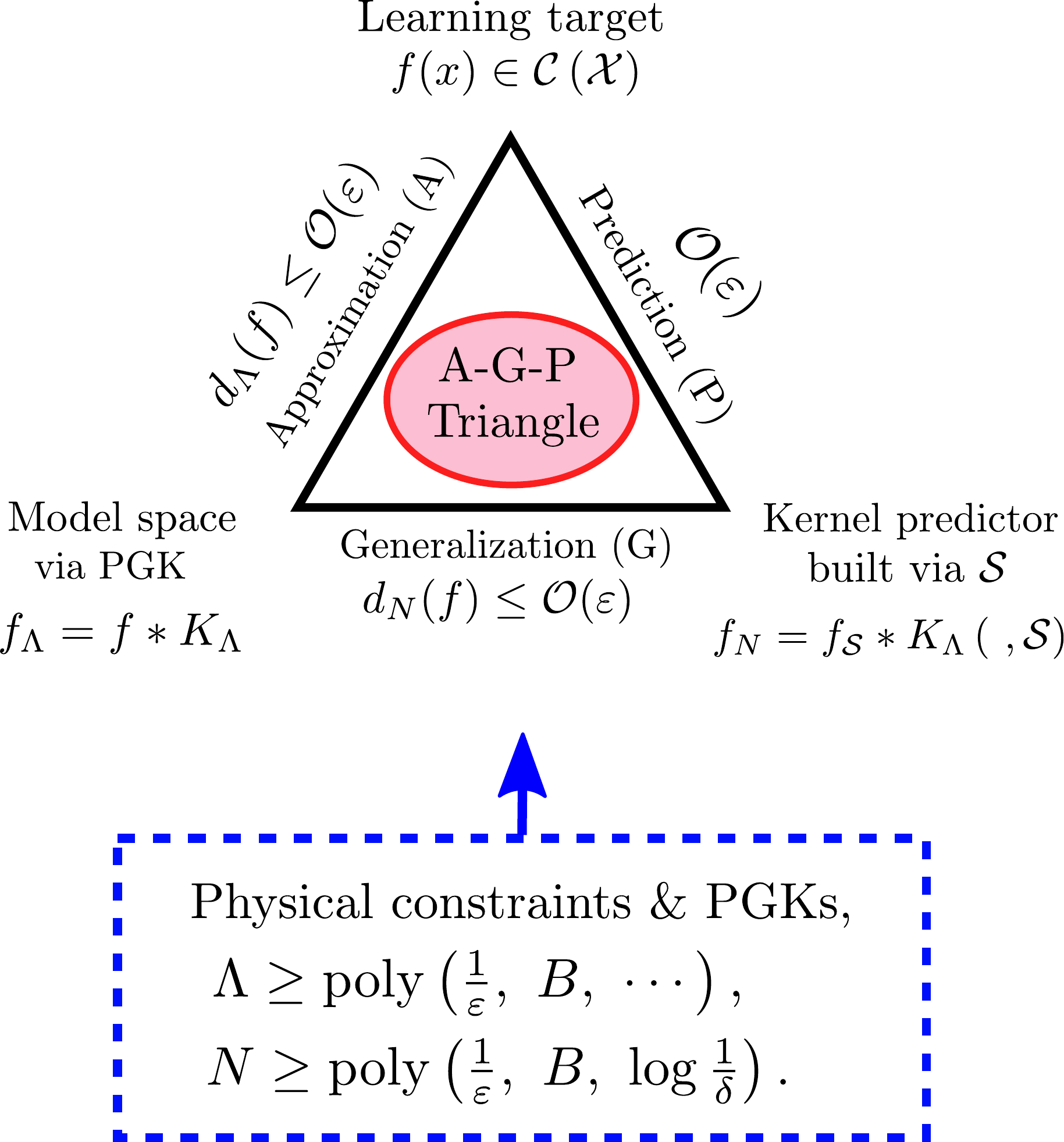}}
\bigskip
\caption{\textbf{Schematics of bounding the prediction error of positive good kernel (PGK) predictors, with polynomial number of samples.} 
The target function to be learned (e.g., entries of the density matrix, expectation values of operators with respect to a target quantum state, etc.), is a continuous function defined on the compact parameter space ${\cal{X}}$, the space of which is denoted by ${\cal{C}} ({\cal{X}})$. In the left bottom of the Approximation-Generalization-Prediction (A-G-P) triangle, the model space is obtained via the convolution (denoted by a $*$) between the function itself and a PGK $K_{\Lambda}$. In the right bottom, the kernel predictor $f_N$ is defined as a discrete convolution over the sample set ${\cal{S}}$ of size $N$. With the physical constraints and proper PGKs, if both the approximation error $d_{\Lambda}(f)$ and the generalization error $d_N(f)$ can be bounded from above up to a small error bound of ${\cal{O}}(\varepsilon)$, in a consistent manner and with polynomial number of samples $N$, then the prediction error can be bounded at ${\cal{O}}(\varepsilon)$, according to the triangle inequality for distances and norms. Note that $B$ is twice the upper bound of the target function $f(x)$ for $x \in {\cal{X}}$ (related to the number of qubits $n$), and $\delta$ is the probability of success with respect to ${\cal{S}}$.
}
\label{fig:ProofSchematic}
\end{figure}

The key idea relies on the concept of good kernels~\cite{FourierBookStein} (We added the positivity condition for good kernels to predict density matrices and their properties. Then the $L_1$-boundness condition for good kernels is automatically satisfied due to the normalization of the kernel). The convolution between a sequence of good kernels and a continuous target function on a closed interval can approximate the function itself up to an arbitrary \emph{uniform} error at ${\cal{O}}(\varepsilon)$ (also known as the ``approximation to the identity")~\cite{FourierBookStein}. On the other hand, the kernel convolution can be approximated as the sample mean over the training set of size $N$ drawn from the parameter space ${\cal{X}}$ and the target $f(x)$ [e.g., see Eq. (1) in the main text], up to a generalization error ${\cal{O}}(\varepsilon)$ and with a probability of success (at least) $1-\delta$. If the two approximation errors can be consistently controlled (for all $x \in {\cal{X}}$) with a polynomially scaling number of samples (relative to $\varepsilon^{-1}$, the qubit number $n$, and $\log \delta^{-1}$), then Theorems $1$ and $2$ follow naturally. Note that we have used the Heine-Cantor theorem which states that a continuous function defined on a compact metric space is uniformly continuous. Here the input space ${\cal{X}}$ is a compact metric space. See Fig.~\ref{fig:ProofSchematic} for a schematic illustration of the main idea. We will elaborate the details of the proof in the following, where we use positive good kernels such as the multi-dimensional Fej\'er kernel and the Gaussian kernel, and the uniform probability distribution. The results can remain unchanged with non-uniform distributions. (This completes the summary of the proof.)

Without loss of generality, we set the input parameter space as $ {\cal{X}} = [-\frac{L}{2}, \frac{L}{2}]^m \subset \mathbb{R}^m$, an $m$-dimensional tensor product of closed intervals of length $L$. The target to learn is a 
\emph{continuous} real function $f: {\cal{X}} \to \mathbb{R}$, defined on a circle of length $L$ in each dimension (i.e., with periodic boundary conditions in the parameter space).
As in the main text, we restrict ourselves to kernels with translational symmetry. For a sequence of kernels $\{ K_{\Lambda} (x)\}_{{\Lambda} \in {\cal{I}}}$ defined on the parameter space ${\cal{X}}$, with $ {\cal{I}}$ the index set of integers by default, we first use the uniform distribution for the probability measure on $\cal{X}$ (the case with non-uniform distributions will be discussed later), i.e., $\mathrm{d}\mu(x) = \mathrm{d} x /L^m$. Then, the properties for positive good kernels read: 

\begin{enumerate}[(I)]
\item Positivity and boundness: $0 \le K_{\Lambda}(x) \le {\cal{O}}\left(\Lambda^{\tau} \right)$ ($\forall x \in {\cal{X}}$), with $\tau$ some positive integer; 

\item Normalization: $\frac{1}{L^m} \int_{x \in {\cal{X}}} K_{\Lambda} (x) \mathrm{d}x = 1$;

\item $\eta$-convergence: For all $0 < \eta \le L$, $\frac{1}{L^m} \int_{x \in {\cal{X}}, \ \| x\|_{2} \ge \eta} | K_{\Lambda} (x) | \mathrm{d}x \le {\cal{O}}\left( \Lambda^{-1}\right)$ as $\Lambda \to \infty$.
\end{enumerate}

We first use a Fej\'er kernel~\cite{FourierBookStein,Pfister2019BoundingMT} as an example of the proof. The rectangular Fej\'er kernel, obtained as the mean of a sequence of rectangular Dirichlet kernels, is given by~\cite{Pfister2019BoundingMT} 
\begin{eqnarray}
\label{eq:Fejer_Kernel_SI}
F_{\Lambda} (x) = \frac{1}{\Lambda^m} \sum_{\mathbf{n} \in V_{\Lambda}} D_{\mathbf{n}} (x), 
\end{eqnarray}
where the kernel index set $V_{\Lambda} = \{ \mathbf{n} = \left( n_1, n_2, \cdots, n_m\right) \in \mathbb{Z}^m : 0 \le n_i < \Lambda, \forall i \in [1, 2, \cdots, m] \}$, and the rectangular Dirichlet kernel reads 
\begin{eqnarray}
\label{eq:rectangular_Dirichlet_SI}
D_{\mathbf{n}} (x) =  \sum_{k_1 = -n_1}^{n_1} \sum_{k_2 = -n_2}^{n_2} \cdots \sum_{k_m = -n_m}^{n_m} e^{2 \pi k \cdot x /L}, 
\end{eqnarray}
where the wave vector $k = \left( k_1, k_2, \cdots, k_m\right)$ takes integer entries.
An equivalent but closed expression for the Fej\'er kernel above is 
\begin{eqnarray}
\label{eq:Fejer_Explicit_SI}
F_{\Lambda} (x) = \frac{1}{\Lambda^m} \prod_{i = 1}^m \frac{\sin^2 \left( \frac{\Lambda \pi}{L} x_i\right)}{\sin^2 \left( \frac{\pi x_i}{L} \right)}, 
\end{eqnarray}
where we have $0 \le F_{\Lambda} (x) \le \Lambda^m$. The closed form (\ref{eq:Fejer_Explicit_SI}) of the rectangular kernel significantly accelerates the numerical convolution compared to the $\ell_2$ kernels, as the latter require to first generate wave vectors whose $\ell_2$-norm does not exceed $\Lambda$, while the cardinality of the wave-vector set can be exponentially large relative to $\Lambda$~\cite{HuangScience2022Manybody}.

\begin{lemma} 
The rectangular Fej\'er kernel is a positive good kernel (PGK) with respect to the uniform distribution.
\end{lemma}

\begin{proof} 
The positivity and boundness, as well as the normalization can be easily verified through the closed expression of the Fej\'er kernel. The last property (III) can be found true in the following section~\ref{sec:BoundApproxError} while bounding the approximation error in (\ref{eq:FejerConv_error}).
\end{proof}

Here the target function $f(x)$ to be learned can be either the entries of the density matrix (real or imaginary parts, respectively), or the average of observables with respect to the quantum state, $f_O (x) = \mathrm{Tr} \left( \hat{O} \rho(x) \right)$, where $\hat{O}$ is a local bounded operator as in the main text. 

\medskip
\textbf{Notation and definition.} 
Because $f(x)$ is assumed to be continuous on a compact metric space ${\cal{X}} \subset \mathbb{R}^m$, according to the Heine-Cantor theorem, it is uniformly continuous. Then, $\forall \varepsilon > 0$, there exists $\eta > 0$ such that, if $\| x - y\|_2 \le \eta$, then $\left| f(x) - f(y) \right| \le \frac{\varepsilon}{4}$ [called $\left( \frac{\varepsilon}{4}, \eta \right)$-continuous]. The function is uniformly continuous, so $\eta$ is independent of $x$, but can depend on $\varepsilon$. Furthermore, for a rich family of continuous functions, $\eta$ can be chosen to be of ${\cal{O}} (\varepsilon^k)$, where $k \ge 0$. For example, for Lipschitz continuous functions, which satisfy the Lipschitz condition 
\begin{equation}
\left| f(x) - f(y)\right| \le C_{L} \| x - y\|_2
\end{equation}
for all $x, y \in {\cal{X}}$, with $C_{L}$ the Lipschitz constant, we have $\eta = \varepsilon/\left(4 C_{L}\right)$. Or more generally, for H\"older continuous functions satisfying  
\begin{equation}
\left| f(x) - f(y)\right| \le C_{H} \| x - y\|_2^{\alpha_H},
\end{equation}
where $C_H$ is a constant and $\alpha_H \in (0, 1]$, we have $k=1/\alpha_H$. As a consequence, a function that scales polynomially in $\eta^{-1}$ is also polynomial in $\varepsilon^{-1}$. In addition, because $f(x)$ is continuous on the closed set ${\cal{X}}$, it is guaranteed to be bounded by $\left| f(x)\right| \le B/2$ for all $x \in {\cal{X}}$, for some positive constant $B$.

\subsection{Bounding the approximation error} 
\label{sec:BoundApproxError}
With positive good kernels, such as the Fej\'er kernel $F_{\Lambda}(x)$, the error from the convolutional approximation of the function is given by
\begin{eqnarray}
\label{eq:FejerConv_error}
d_{\Lambda} (f) &=& \left| f*F_{\Lambda}(x) - f(x)\right| \nonumber \\ 
&=& \frac{1}{L^m} \left| \int_{{\cal{X}}}\left[f(x-y) - f(x)\right] F_{\Lambda}(y) \mathrm{d}y\right| 
\nonumber \\ 
&\le& \frac{1}{L^m} \int_{{\cal{X}}}\left| f(x-y) - f(x)\right| \times \left| F_{\Lambda}(y) \right| \mathrm{d}y  \nonumber \\
&=& \frac{1}{L^m} \int_{{y \in \cal{X}}, \| y\|_2 \le \eta}  \left| f(x-y) - f(x)\right| \times \left| F_{\Lambda}(y) \right| \mathrm{d}y  \nonumber \\ 
&+&  \frac{1}{L^m} \int_{{y \in \cal{X}}, \| y\|_2 \ge \eta}  \left| f(x-y) - f(x)\right| \times \left| F_{\Lambda}(y) \right| \mathrm{d}y  \nonumber \\ 
&\le& \frac{\varepsilon}{4} + B \left( \frac{C}{\Lambda}\right)^m  \le \frac{\varepsilon}{2}, 
\end{eqnarray}
where 
\begin{equation}
\label{eq:DefinitionC}
C = \frac{4m L^2} {\pi^2 \eta^2}.
\end{equation}
The second equation comes from the symmetry property of the kernel convolution on ${\cal{X}}$, $f*F_{\Lambda}(x) = F_{\Lambda}*f(x)$, as well as the normalizaiton property of the Fej\'er kernel. In the third equation, the first integration is upper bounded by $\varepsilon/4$ because $f$ is $\left(\frac{\varepsilon}{4}, \eta \right)$-continuous and because the kernel is normalized. The bound for the second term in the third equation is obtained based on the fact that the integration over the region with $\| y\|_2 \ge \eta$ does not exceed the integration over the the region with $\| y\|_{\infty} \ge \eta/\sqrt{m}$, where the latter is a multi-dimensional integration over a cubic region that can be factorized into the product of the integrations in each dimension. In each factor dimension, the Fej\'er kernel is bounded from above by 
\begin{equation}
\frac{1}{\Lambda \sin^2\left[\pi \eta/(\sqrt{m} L) \right]} \le \frac{4m L^2} {\pi^2 \Lambda \eta^2 },
\end{equation}
by using $\sin x \ge x/2$ while $0 < x < \pi/3$ is a small number. Then the second inequality holds. \emph{This line of argument also serves as a proof for the statement that the multi-dimensional Fej\'er kernel satisfies the property $(\mathrm{III})$ of good kernels (see above).} 
Finally, if we set 
\begin{eqnarray}
\label{eq:Bound_Lambda}
\Lambda^m \ge \frac{4 B C^m}{\varepsilon}, 
\end{eqnarray}
it suffices to bound the total approximation error $d_{\Lambda}(f)$ up to $\varepsilon/2$, as in the last inequality.

\subsection{Bounding the generalization error}
\label{sec:BoundGeneralizationError}

The Fej\'er kernel estimator of the density matrix is given by 
\begin{eqnarray}
\label{eq:Fejer_DM_SI}
\sigma_N (x) = \frac{1}{N} \sum_{i=1}^N F_{\Lambda} (x-x_i) \rho(x_i).
\end{eqnarray}
By taking the target $f(x)$ to be the entries of the density matrix or the state properties $f_O(x)$, and with the corresponding sample set ${\cal{S}} = \{x_i, f(x_i) \}_{i=1}^N$ drawn uniformly from ${\cal{X}}$ and the function space, Eq.~(\ref{eq:Fejer_DM_SI}) is a sample mean of the kernel convolution $f*F_{\Lambda}(x)$. Equivalently, denoting a sequence of independent random variables as $\{Z_i (x) = F_{\Lambda} (x-x_i) f(x_i)\}_{i=1}^{N}$, which is bounded by $|Z_i| \le \Lambda^m B/2$, we have 
\begin{eqnarray}
\label{eq:Expectation_Z}
\underset{x_i \sim \mu \left({\cal{X}}\right)} {\mathbb{E}} \frac{1}{N} \sum_{i=1}^N Z_i (x)
&=& \frac{1}{L^m} \int_{{\cal{X}}} F_{\Lambda}(x-y) f(y) \mathrm{d}y \nonumber \\
&=& f*F_{\Lambda}(x).
\end{eqnarray}
As in Eq.~(\ref{eq:Fejer_DM_SI}), the Fej\'er kernel estimator can be written as $\hat{f}_N (x) = \frac{1}{N} \sum_{i=1}^N Z_i (x)$. The McDiarmid's inequality~\cite{FoundationsMLBook} leads to  
\begin{eqnarray}
\label{eq:Heoffding}
\mathrm{Prob.} \left( \left|\hat{f}_N (x) - f*F_{\Lambda}(x) \right| \ge \frac{\varepsilon}{2} \right) \le 2 \exp \left( - \frac{N \varepsilon^2}{2 \Lambda^{2m} B^2}\right). \nonumber \\
\end{eqnarray}
Therefore, if the number of samples is set to be 
\begin{eqnarray}
\label{eq:BoundGeneralizationErrorN}
N \ge \frac{2B^2 \Lambda^{2m}}{\varepsilon^2} \log \frac{2}{\delta},
\end{eqnarray}
then we have the generalization error 
\begin{eqnarray}
d_N (f) = \left|\hat{f}_N (x) - f*F_{\Lambda}(x) \right| \le \frac{\varepsilon}{2},
\end{eqnarray}
with probability at least $1-\delta$ ($\delta \in (0, 1)$).

\subsection{Bounding the supremum-norm prediction error with polynomial sample complexity}
\label{sec:BoundPredictionError}
Combing Eqs.~(\ref{eq:Bound_Lambda}) and (\ref{eq:BoundGeneralizationErrorN}), we can finally obtain the prediction error of the Fej\'er kernel estimator as (see Fig.~\ref{fig:ProofSchematic} for an illustration)
\begin{eqnarray}
\left|\hat{f}_N (x) - f(x) \right| \le d_N (f) + d_{\Lambda} (f) \le \varepsilon, \ \ \ \forall x \in {\cal{X}}, \nonumber \\
\end{eqnarray}
with probability at least $1 - \delta$, by using a sample set of size
\begin{eqnarray}
\label{eq:N_Fejer_SI}
N &\ge& \frac{32 B^4 C^{2m} }{\varepsilon^4} \log \frac{2}{\delta} \nonumber \\ 
&=& \mathrm{poly} \left(\frac{1}{\varepsilon}, \ B, \ \log \frac{1}{\delta}\right),
\end{eqnarray}
provided that $m$ is a finite constant, the scaling of which is not a concern. The polynomial scaling of $N$ in $\varepsilon^{-1}$ is due to the fact that $C = {\cal{O}}(\eta^{-2}) = {\cal{O}}(\varepsilon^{-2k})$, as in (\ref{eq:DefinitionC}) and in the \textbf{Notation and definition} of Sec.~\ref{sec:MainIdeas}. 

\bigskip
Now we can apply (\ref{eq:N_Fejer_SI}) to the quantum-state properties learning as in Theorem 2 of the main text, where the target function $f_O (x) = \mathrm{Tr} \left[ O \rho(x)\right] = \sum_{i=1}^{M} f_i(x)$, is a continuous function on ${\cal{X}}$ given by the average of a local operator $O = \sum_{i=1}^{M} O_i$, as in the main text. The bound of $f_O(x)$ here is $B = {\cal{O}} (M) = \mathrm{poly}(n)$, provided that $\left| f_i (x) \right| \le {\cal{O}} (1)$ as in the main text.

\bigskip
So far we have focused on approximating a single function. Next we consider to simultaneously bound the approximation errors up to $\varepsilon$ for a set of continuous functions $\{f_i (x) \}_{i=1}^M$ with $x \in \cal{X}$, with the Fej\'er kernel predictor as above. This situation applies to, for example, learning the overall density matrix as in Theorem 1 and learning the quantum-state properties with locality assumptions as in Theorem 2. 

Define the event ${\cal{A}}_i$ to be $\left|\hat{f}_{i N} (x) - f_i*F_{\Lambda}(x) \right| \le \frac{\varepsilon}{2}$, then 
\begin{eqnarray}
\label{eq:ProbIntersection}
&\mathrm{Prob.}& \left( \left|\hat{f}_{i N} (x) - f_i*F_{\Lambda}(x) \right| \le \frac{\varepsilon}{2}, \ 
\forall i \in [M] \right) \nonumber \\ &=& \mathrm{Prob.} \left( \bigcap_{i=1}^M {\cal{A}}_i \right) \nonumber \\ 
&=& 1 - \mathrm{Prob.} \left( \bigcup_{i=1}^M {\cal{A}}^c_i \right), 
\end{eqnarray}
where ${\cal{A}}^c_i = : \left|\hat{f}_{i N} (x) - f_i*F_{\Lambda}(x) \right| \ge \frac{\varepsilon}{2}$ is the complementary event of ${\cal{A}}_i$. Also, we have the following relation, 
\begin{eqnarray}
\label{eq:ProbUnion}
\mathrm{Prob.} \left( \bigcup_{i=1}^M {\cal{A}}^c_i \right) &\le& \sum_{i=1}^M \mathrm{Prob.} \left( {\cal{A}}^c_i \right) \nonumber \\ 
&\le& M \underset{i \in [M]}{\max} \mathrm{Prob.} \left( {\cal{A}}^c_i \right) \\
&\le& 2M \underset{i \in [M]}{\max} \exp \left( - \frac{N \varepsilon^2}{2 \Lambda_i^{2m} B_i^2}\right), \nonumber 
\end{eqnarray}
where the first inequality is from the fact that, $\mathrm{Prob.} \left( A_1 \cup A_2 \right) = \mathrm{Prob.} \left( A_1 \right) + \mathrm{Prob.} \left(A_2 \right) - \mathrm{Prob.} \left( A_1 \cap A_2 \right) \le \mathrm{Prob.} \left( A_1 \right) + \mathrm{Prob.} \left(A_2 \right)$, and the last inequality is derived from (\ref{eq:Heoffding}). Lastly, by combining Eqs.~(\ref{eq:Bound_Lambda}, 
\ref{eq:ProbIntersection}, \ref{eq:ProbUnion}), we obtain that, with success probability at least $1-\delta$, all of the functions $f_i(x)$ can be approximated by the kernel predictor up to a maximal error $\varepsilon$, if the kernel parameter and the number of samples are set to be
\begin{eqnarray}
\Lambda^m \ge \frac{4 B_{\mathrm{max}} C_{\mathrm{max}}^m}{\varepsilon}, 
\end{eqnarray}
and
\begin{eqnarray}
\label{eq:N_Fejer_BCmax}
N &\ge& \frac{32 B_{\mathrm{max}}^4 C_{\mathrm{max}}^{2m} }{\varepsilon^4} \log \frac{2M}{\delta} 
\nonumber \\ 
&=& \mathrm{poly} \left(\frac{1}{\varepsilon}, \ B_{\mathrm{max}}, \ \log \frac{M}{\delta}\right),
\end{eqnarray}
respectively, where the respective maximal values of $B$ and $C$ are obtained by taking $i$ running from $1$ to $M$. Note that for the density matrix learning in Theorem 1, $B_{\mathrm{max}} = \| \rho\|_{\ell_{\infty}} \le 1$ and $M = {\cal{O}} (2^n)$; while for the quantum-state properties learning with strong locality in Theorem 2, $B_{\mathrm{max}} = {\cal{O}} (1)$ and $M = \mathrm{poly} (n)$. 

\bigskip
In addition, the normalization property of good kernels is particularly important for the density matrix learning problem, in keeping the trace of $\sigma_N(x)$ to be (approximately) one for all $x \in {\cal{X}}$. The Hoeffding inequality (or more generally, the McDiarmid's inequality)~\cite{FoundationsMLBook} ensures that the trace of the predicted density matrix, given by $\mathrm{Tr} \sigma_N (x) = \frac{1}{N} \sum_{i=1}^N F_{\Lambda} (x - x_i)$, lies in the vicinity of one with high probability. In concrete, we have 
\begin{eqnarray}
\label{eq:Trace_Norm}
\mathrm{Prob.} \left( \left|\mathrm{Tr} \sigma_N (x) - 1 \right| \ge \varepsilon \right) \le 2 \exp \left( - \frac{2 N \varepsilon^2}{\Lambda^{2m}}\right). \nonumber \\
\end{eqnarray}
Therefore, by taking $N$ in Eq.~(\ref{eq:N_Fejer_BCmax}), it suffices to set $\left|\mathrm{Tr} \sigma_N (x) - 1 \right| \le \varepsilon$ with probability at least $1-\delta$. 

\bigskip \bigskip
So far we have proved the statements in Theorems $1$ and $2$ in the main text with an example of PGKs, the rectangular Fej\'er kernel. Alternatively, one can use a Gaussian kernel as another choice of the PGKs, in which case the index set of the kernel is not restricted to positive integers. Then the Fej\'er kernel in Eq.~(\ref{eq:Fejer_DM_SI}) can be replaced by a Gaussian kernel, which is given by 
\begin{eqnarray}
\label{eq:Gaussian}
K_h (x) = C_h \mathrm{exp} \left( - \frac{\| x\|_{\ell_2}^2}{h} \right),
\end{eqnarray}
where $x \in {\cal{X}}$ and $C_h$ is the normalization coefficient. Note that here $1/\sqrt{h}$ plays the same role as $\Lambda$ in the Fej\'er kernel. We set $h \ll L/2$, then $C_h = \left( L/\left( \sqrt{\pi h}\right)\right)^m$. We extend the Gaussian kernel to be defined periodically outside ${\cal{X}}$, i.e., to be defined on a circle of length $L$ in each dimension as in the case of the Fej\'er kernel. Then, following a similar procedure as above, it can be straightforwardly shown that, for $m \ge 2$, the required number of samples for $\left|\hat{f}_N (x) - f(x) \right| \le d_N (f) + d_{\Lambda} (f) \le \varepsilon$, with confidence at least $1-\delta$ ($\delta \in (0, 1)$), is 
\begin{eqnarray}
\label{eq:N_Gaussian_SI}
N &\ge& \frac{2 B^2 C_g^{m} }{\varepsilon^2} \log \frac{2}{\delta} \nonumber \\
&=& \mathrm{poly} \left(\frac{1}{\varepsilon}, \ B, \ \log \frac{1}{\delta}\right),
\end{eqnarray}
where $C_g = \frac{m L^2}{\pi \eta^2} \log \frac{2mB}{\varepsilon}$. So the result is similar to that with a Fej\'er kernel in (\ref{eq:N_Fejer_SI}). Other examples of PGKs, which may have better performance, will be left for future work. 

\subsection{Non-uniform distributions}
\label{sec:NonUniformDistr}
We have demonstrated a polynomial sample complexity with a uniform probability density, $\mathrm{d}\mu(x) = \mathrm{d} x /L^m$. We recall that~\cite{FoundationsMLBook} the framework of the Probably Approximately Correct (PAC) learning also requires that the polynomial sample complexity is distribution-free. So it is meaningful to explore other practical probability densities, $\varrho(x)$, with which the sample $x_i$ is drawn from ${\cal{X}}$, denoted by $\mathrm{d}\mu(x) = \varrho(x) \mathrm{d} x$. Start from the above rectangular Fej\'er kernel as an example, one can use a weighted Fej\'er kernel, given by 
\begin{equation}
\tilde{F}_{\Lambda} (x) = \omega(x) F_{\Lambda}(x),
\end{equation} 
where $\omega(x) = \varrho_0(x)/\varrho(x)$, with $\varrho_0(x) = 1/L^m$ the uniform density distribution. Then the weighted Fej\'er kernel is a PGK when the probability density $\varrho(x)$ is used, and serves as the kernel used for a similar derivation leading to a polynomial sample complexity as in the above proving procedure, provided that $\tilde{F}_{\Lambda} (x)$ is also upper bounded by ${\cal{O}}\left( \Lambda^m\right)$ for $x \in {\cal{X}}$, as in the case of the Fej\'er kernel. While this is one of the possible methods to construct PGKs for non-uniform distributions, other approaches may exist.

\bigskip
In summary, by combining above Secs.~\ref{sec:MainIdeas}-\ref{sec:NonUniformDistr}, we complete the unified proof for Theorems $1$ and $2$ in the main text.

\medskip
\subsection{Comparison with existing results}
\label{sec:ComparisonExistingResults} 

Again, we compare the number of samples above with the exponential scaling for learning the ground-state properties $f_O(x)$ in~\cite{HuangScience2022Manybody}, with more mathematical details compared to that in the main text. The target function defined on the compact parameter space (with $L=2$) in~\cite{HuangScience2022Manybody} is assumed to have a bounded first-order derivative, which means it satisfies the Lipschitz condition mentioned earlier, 
\begin{equation}
\left| f_O(x) - f_O(y)\right| \le C_{L} \| x - y\|_2, 
\end{equation}
for all $x, y \in {\cal{X}}$. This is a result of the mean value theorem, where the first-order derivative of the function is upper bounded by $C_L$ here. As a consequence, we have $\eta = \varepsilon/\left( 4 C_L\right)$ (or equivalently, $k = 1$). Also, note that Ref.~\cite{HuangScience2022Manybody} uses an error metric of averaged square distance between the predicted and target functions, while here we used the supremum-norm distance bounded by $\varepsilon$. With the same notation, the required number of samples in~\cite{HuangScience2022Manybody} is of the order (with $C_L = 1$)
\begin{equation}
N_0 =\frac{B^2}{\varepsilon^2} (2m + 1)^{\frac{1}{\varepsilon^2}}.
\end{equation}
Note that $\left| f_O (x)\right| \le B/2 = {\cal{O}}(M)$ by assuming that $\|O_i \|_{\infty} = {\cal{O}}(1)$, and $M = {\cal{O}}(n^q)$ for the $q$-local operator $O$. 

\medskip
For example, taking $m = 2$, the qubit number $n=100$, $q=1$, and $\varepsilon = 0.1$, one can easily obtain 
\medskip
\begin{equation}
\frac{N}{N_0} \approx 10^{-48},
\end{equation}
with a Fej\'er kernel, where the required number of samples $N$ is given by (\ref{eq:N_Fejer_SI}), and 
\begin{equation}
\frac{N}{N_0} \approx 10^{-61},
\end{equation}
with a Gaussian kernel, where the required number of samples $N$ is given by (\ref{eq:N_Gaussian_SI}), respectively. \emph{We see that when $m$ is a small constant, very significant reductions in the provable sample complexity can be achieved here with the PGKs, as in the main text}.

\medskip
\section{Bounded smooth parametrization implies continuous density matrices}
\label{sec:SmoothParametrization}

In Theorem 1, we assume that the entries of the density matrix are continuous over the parameter space. In this section, we show the continuity of the entries of the density matrix, provided that the operator norm of the gradient of the parametrized density operator is uniformly upper bounded. This holds for ground states of gapped Hamiltonians and other smoothly parametrized models. 

\begin{lemma} 
\label{lemma:DM_Conti}
For a family of density operators $\rho(x)$ smoothly parametrized by $x$, if the gradient is upper bounded uniformly in the operator norm, i.e., $\|\partial_x \rho(x) \|_{\infty} \le C_0$, then each entry of the density matrix is continuous over the parameter space.
\end{lemma}

\begin{proof} 
Because $ \partial_x \rho(x)$ is a Hermitian operator, it has a spectral decomposition given by 
\begin{eqnarray}
\label{eq:GradientRho_spec_decom}
 \partial_x \rho(x) = \sum_{k =1} \lambda_k |e_k\rangle \langle e_k |,
\end{eqnarray}
where $\lambda_k$ is the eigenvalue and $|e_k\rangle$ is the corresponding eigenvector (which both may be $x$-dependent).

As a result, the gradient of the entry of the density matrix is bounded by 
\begin{eqnarray}
\label{eq:GradientRho_spec_decom}
\left| \partial_x \rho_{i j}(x) \right| &=& \left| \sum_{k =1} \lambda_k \langle i| e_k\rangle \langle e_k | j \rangle \right| \nonumber \\ 
&\le& \sum_{k =1} \left| \lambda_k \right| \ \left| \langle i| e_k\rangle \langle e_k | j \rangle \right| 
\nonumber \\ 
&\le& \|\partial_x \rho(x) \|_{\infty} \sum_{k =1} \left| \langle i| e_k\rangle \langle e_k | j \rangle \right| \nonumber \\ 
&\le& C_0 \left(\sum_{k=1} \left| \langle i | e_k \rangle\right|^2 \sum_{k=1} \left| \langle j | e_k \rangle\right|^2 \right)^{\frac{1}{2}} \nonumber \\ 
&\le& C_0. 
\end{eqnarray}
where $|i\rangle$ is a parameter-independent basis of the Hilbert space, and in the second last inequality we have used the Cauchy-Schwartz inequality and the normalization condition $\langle i | i \rangle = \sum_{k=1} \left| \langle i | e_k \rangle\right|^2 = 1 \ \forall i$.

Because the gradient of the entry is uniformly bounded by $C_0$ as above, it is a Lipschitz continuous function over the parameter space, which can be directly obtained via the mean value theorem. It can be directly generalized to the multiparameter case, provided that Lemma~\ref{lemma:DM_Conti} holds in each dimension of the parameter space.
\end{proof}

\medskip
In the following we list a few examples where the gradient of the density operator can be uniformly bounded.

\medskip
\subsection{Ground states of gapped Hamiltonians} The gradient of the ground-state density operator for a local Hamiltonian $H(x)$ ($0 \le x \le 1$) with a uniform constant gap $\gamma > 0$ is given by~\cite{BachmannCMP2011}
\begin{eqnarray}
\label{eq:GradientRho_GS_gapped}
\partial_x \rho(x) &=& i \left[\rho(x), D(x) \right], \nonumber\\  
\mathrm{with} \ \ D(x) &=& \int_{-\infty}^{+\infty} \mathrm{d}t \ W_{\gamma}(t) \ \mathrm{e}^{it H(x)} 
\, \partial_x H(x) \, \mathrm{e}^{-it H(x)}, \nonumber \\
\end{eqnarray}
where $W_{\gamma}(t) \in L^1 (\mathbb{R})$ is a decaying weight function satisfying $\left\| W_{\gamma} \right\|_1 \le c_0/\gamma$, for some positive constant $c_0$~\cite{BachmannCMP2011}.

It is straightforward to show that 
\begin{eqnarray}
\left \| \partial_x \rho(x) \right \|_{\infty} &\le& 2 \left \| D(x) \right \|_{\infty}  \left \| \rho(x) \right \|_{\infty}\nonumber \\ 
&\le& 2 \int_{-\infty}^{+\infty} \mathrm{d}t \ \left|W_{\gamma}(t) \right| \left \| \partial_x H(x) \right \|_{\infty} \nonumber \\
&\le& \frac{2c_0}{\gamma} \left \| \partial_x H(x) \right \|_{\infty}.
\end{eqnarray}
Note that we have used that $\left \| \rho(x) \right \|_{\infty} \le 1$ and that the operator norm is invariant under unitary transformations in the second inequality. Therefore, the gradient of the density matrix above can be bounded provided that each component of the local Hamiltonian $H(x) = \sum_j h_j(x)$ has a bounded gradient in the operator norm, i.e., $\|\partial_x h_j(x) \|_{\infty} \le \mathrm{constant}$. 

\medskip
\subsection{States smoothly parametrized by Kraus operators}
In the Kraus representation~\cite{NielsenChuang}, starting from a parameter independent initial state $\rho_0$, the parametrized state (which can be a pure or a mixed state) is given by 
\begin{equation}
\rho(x) = \sum_i K_i (x) \rho_0 K_i^{\dagger} (x), 
\end{equation}
where $K_i (x)$ are the Kraus operators satisfying $\sum_i K_i (x) K_i^{\dagger} (x) = \mathrm{Id \ (identity)}$. The gradient of the quantum states satisfies 
\begin{eqnarray}
\label{eq:GradientRho_Kraus}
\left \| \partial_x \rho(x) \right \|_{\infty} \le 2 \sum_i \left \| K_i (x) \right \|_{\infty} \left \| \partial_xK_i (x) \right \|_{\infty},
\end{eqnarray}
where we have used the subadditivity and submultiplicity of the operator norm, and $\left \| \rho_0 \right \|_{\infty} \le 1$.

Therefore, the gradient of $\rho(x)$ can be uniformly upper bounded, as long as the operator norms of $K_i (x)$ and $\partial_x K_i (x)$ ($\forall i$) for the physical process are uniformly bounded from above.

\medskip
\subsection{Quantum parameter estimation} 
In quantum parameter estimation or quantum metrology, quantum Fisher information plays a central role in characterizing the fundamental precision limit (e.g., for a recent review, see~\cite{LiuJPAReview2020}; see also~\cite{Helstrom,Holevo,MaPhysRep2011,ChePRA2019}). The quantum Fisher information with respect to the single parameter $x$ for the quantum state $\rho(x)$ reads 
\begin{equation}
F = \mathrm{Tr} \left[\rho(x) {\cal{L}}^2 \right], 
\end{equation}
where ${\cal{L}}$ is the symmetric logarithmic derivative (SLD) operator given by 
\begin{eqnarray}
\label{eq:GradientRho_SLD}
\partial_x \rho(x) = \frac{1}{2} \left[ \rho(x) {\cal{L}}+ {\cal{L}} \rho(x) \right].
\end{eqnarray}

Consequently, we can obtain that $\left \| \partial_x \rho(x) \right \|_{\infty} \le \left \| {\cal{L}} \right \|_{\infty}$. This condition holds in quantum metrological schemes with finite amount of quantum Fisher information (e.g., without quantum phase transitions in the parameter range of interest).

\section{Review of reproducing kernel Hilbert space (RKHS)}
\label{sec:ReviewRKHS}

Here we provide a brief review of key ideas and main results for the RKHS~\cite{SVMBook2008,arxiv2101.11020}. 

\begin{definition} 
Let ${\cal{X}}$ be a compact space and $\cal{H}$ be a Hilbert space of functions $f: {\cal{X}} \to \mathbb{R}$ with the inner product $\left\langle \cdot, \cdot \right \rangle_{\cal{H}}$, a kernel associated with $\cal{H}$ is defined as $K: {\cal{X}} \times {\cal{X}} \to \mathbb{R}$. If $K ({\cal{X}}, {\cal{X}})$ is symmetric and positive definite, and has the following reproducing property, i.e., for all functions $f \in {\cal{H}}$, $f(x) = \left\langle f, K(x, )\right \rangle_{\cal{H}}$, then $\cal{H}$ is called the reproducing kernel Hilbert space (RKHS) of $K$, denoted by ${\cal{H}}_K$.
\end{definition}

Properties of RKHS:

\begin{enumerate}[(1)]
\item $\forall x \in {\cal{X}}$, $K(x, ) \in {\cal{H}}_K$;

\item $K(x, x^{\prime}) = \left \langle K(x, ), K(x^{\prime}, ) \right\rangle_{{\cal{H}}_K}$, which naturally follows from (1) and the reproducing property of RKHS;

\item The linear evaluation operator ${\cal{E}}_x$ is bounded, i.e.,
\begin{eqnarray}
\label{eq:eval_operator}
|{\cal{E}}_x f| &=& |f(x)| = | \langle f, K(x, )\rangle_{{\cal{H}}_K} |  \nonumber \\
&\le& \sqrt{K(x, x)} \|f\|_{{\cal{H}}_K},
\end{eqnarray}
where the second equation follows from the reproducing property and the last inequality results from the Cauchy-Schwartz inequality.
\end{enumerate}

\medskip
\textbf{Theorem (Moore-Aronszajn):} \ \
Let $K: {\cal{X}} \times {\cal{X}} \to \mathbb{R}$ be a positive definite kernel. There is a unique RKHS ${\cal{H}}_K$ with the reproducing kernel $K$.

\medskip

\textbf{Universal kernels:} \ \ Let $C({\cal{X}})$ be the space of bounded continuous functions on ${\cal{X}}$. A kernel $K({\cal{X}}, {\cal{X}})$ is called universal if the RKHS ${\cal{H}}_K$ is dense in $C({\cal{X}})$, i.e., $\forall f \in C({\cal{X}})$ and $\forall \varepsilon > 0$, there is a $f^{*} \in {\cal{H}}_K$ such that $\|f(x) - f^*(x) \|_{\infty} \le \varepsilon$.

\medskip
\textbf{Representer theorem:} \ \
The optimal function $\hat{f} $ in the RKHS that minimizes the functional-norm regularized empirical loss on the training set ${\cal{S}} = \{x_i, y_i \}_{i=1}^N$:
\begin{eqnarray}
\hat{f} = \underset{f \in {\cal{H}}_K}{\arg \min} \, \frac{1}{N}\sum_{i=1}^N \ell(f(x_i), y_i) 
+ \lambda \| f \|_{{\cal{H}}_K},
\end{eqnarray}
takes the form
\begin{eqnarray}
\hat{f}(x) = \sum_{i=1}^N \alpha_i K(x, x_i),
\end{eqnarray}
where the dual variables $\{\alpha_i\} \subset \mathbb{R}$ are determined by minimizing the specific empirical loss function, and $\lambda >0$ is a regularization parameter to avoid the overfitting. This is particularly useful for reducing the complexity of the problem, by transforming the optimization over an infinite-dimensional Hilbert space into that over a finite set of real variables.

\bigskip
\section{Learning quantum-state properties in reproducing kernel Hilbert space and generalization error}
\label{sec:GSPinRKHS}

Here we present the generalization error bound for ansatz functions from the reproducing kernel Hilbert space (RKHS)~\cite{FoundationsMLBook}.
Given spaces ${\cal{X}} \subset \mathbb{R}^m$ and $Y \subset \mathbb{R}$, and a joint distribution $\cal{D}$ over the joint space, where $(x, y) \sim {\cal{D}}$ for $x \in {\cal{X}}$ and $y \in Y$, we are aiming at learning the continuous target function $f_O: {\cal{X}} \to Y$, provided a sample set 
\begin{equation}
{\cal{S}} = \{z_i = (x_i, y_i)\}_{i=1}^N = \{z_i = \left(x_i, f_O(x_i) \right)\}_{i=1}^N 
\end{equation}
drawn from the distribution $\cal{D}$. An optimal function $\hat{f}_O \in {{\cal{H}}_K}$ can be obtained by minimizing an empirical loss defined on the sample set, then the estimator $\hat{f}_O$ will be applied to the whole space ${\cal{X}}$ for predictions.

If we take the loss to be 
\begin{equation}
\ell \left( \hat{f}_O(x), y \right) = \left| \hat{f}_O(x) - f_O (x) \right|, 
\end{equation}
the empirical error is given by 
\begin{eqnarray}
{\cal{E}}_t = \frac{1}{N} \sum_{i=1}^N \left|\hat{f}_O(x_i) - f_O(x_i) \right|,
\end{eqnarray}
and the expected prediction error is 
\begin{eqnarray}
{\cal{E}}_p &=& \underset{(x, y) \sim {\cal{D}}}{\mathbb{E}} \left| \hat{f}(x) - y  \right| \nonumber \\
&=& \underset{x}{\mathbb{E}} \left| \hat{f}_O(x) - f_O(x) \right|.
\end{eqnarray}

In statistical learning theory~\cite{FoundationsMLBook,SVMBook2008}, the expected prediction error is upper bounded by the empirical error plus a term proportional to $1/ \sqrt{N}$. For the model with the RKHS ${\cal{H}}_K$, it can be written as follows:

With probability at least $1 - \delta$ ($\delta \in (0, 1)$),
\begin{eqnarray}
\label{eq:PredictionErrorRKHS_supp}
{\cal{E}}_p \le {\cal{E}}_t + {\cal{R}_{{\cal{S}}}} \left( {\cal{L}} \circ {{\cal{H}}_K}\right) + 3\beta \sqrt{\frac{\log \frac{2}{\delta}}{2N}},
\end{eqnarray}
where 
\begin{eqnarray}
{\cal{R}_{{\cal{S}}}} \left( {\cal{L}} \circ {{\cal{H}}_K}\right) \le {\cal{R}_{{\cal{S}}}}
\left( {{\cal{H}}_K}\right) = \frac{2\lambda_f \sqrt{\mathrm{tr}K}}{N}
\end{eqnarray}
is the empirical Rademacher complexity of the loss space, and provided that $\|\hat{f}_O\|_{{\cal{H}}_K} \le \lambda_f$, and $0 \le \ell \left( \hat{f}_O(x), f_O(x) \right) \le \beta$ for all $x \in {\cal{X}}$. Note that we have used the property of the Rademacher complexity that 
\begin{eqnarray}
{\cal{R}_{{\cal{S}}}} \left( {\cal{L}} \circ {{\cal{H}}_K}\right) \le C_L {\cal{R}_{{\cal{S}}}} \left( {{\cal{H}}_K}\right),
\end{eqnarray}
provided that $\ell(\cdot, y) \in {\cal{L}}$ is a $C_L$-Lipschitz continuous function ($C_L$ = 1 for the loss we use here). The trace of the kernel defined on the training set satisfies $\mathrm{tr}K \le N R^2$ provided that $K(x, x) \le R^2$. 

With the PGK estimator, 
\begin{equation}
\hat{f}_O (x) = \frac{1}{N} \sum_{i=1}^N K_{\Lambda} (x - x_i) f_O (x_i), 
\end{equation}
the empirical error can be approximated as (when $\Lambda$ is a large number, the PGK plays a role of the Dirac delta function) 
\begin{eqnarray}
{\cal{E}}_t \approx \frac{1}{N} \sum_{i=1}^N \left| \frac{1}{N} \sum_{j=1}^N K_{\Lambda}(x_i - x_j) - 1\right| \times \left|f_O (x_i)\right|, \nonumber \\
\end{eqnarray}
which can be upper bounded at ${\cal{O}}(\varepsilon)$ with large probability, given a polynomial $N$ as in (\ref{eq:N_Fejer_BCmax}), with which $\mathrm{Tr} \sigma_N(x_i)$ is close to one up to a ${\cal{O}}(\varepsilon)$ error; Also, note that $f_O (x)$ is bounded.
Once ${\cal{E}}_t$ is minimized up to ${\cal{O}}(\varepsilon)$, one needs to bound the second plus the third terms on the right-hand side of (\ref{eq:PredictionErrorRKHS_supp}) also up to ${\cal{O}}(\varepsilon)$, which will result in ${\cal{E}}_p \le {\cal{O}}(\varepsilon)$, to facilitate a good generalization and prediction. Then we get the minimal number of samples required for this purpose
\begin{eqnarray}
\label{eq:SampleComplexity_SI}
N &\ge& \frac{1}{\varepsilon^2} \left( 8\lambda_f R \sqrt{\frac{\log \frac{2}{\delta}}{2}} \right)^2 
\nonumber \\ 
&=& \mathrm{poly} \left( \frac{1}{\varepsilon}, \ B,\ \log\frac{1}{\delta}\right),
\end{eqnarray}
as in the main text, \emph{which indicates an efficient sample complexity}. Note that with the reproducing property for the function $\hat{f}_O(x)$, one obtains that 
\begin{eqnarray}
\ell \left( \hat{f}_O(x), f_O(x) \right) &\le& \left|\hat{f}_O(x)\right| + \left| f_O(x) \right| \nonumber \\ 
&\le& \lambda_f R + \frac{B}{2} \\ 
&\le& 2\lambda_f R =: \beta, \nonumber
\end{eqnarray}
where the second inequality is obtained from (\ref{eq:eval_operator}), and the last inequality is because $R$ is a large number such that $\lambda_f R \ge B/2$. Furthermore, if the dual variables are set to be 
\begin{equation}
\alpha_i \approx \frac{f_O(x_i)}{N},
\end{equation}
as in the PGK estimator in the main text, we have $\lambda_f = B R/2$. Also, to obtain (\ref{eq:SampleComplexity_SI}), we have assumed $\log \frac{2}{\delta} \ge 2$ for a small value of $\delta$. The result in (\ref{eq:SampleComplexity_SI}) is consistent with that obtained from (\ref{eq:Heoffding}). 
Note that due to the normalization of the kernel, $R$ can be in general dependent on $m$ exponentially. This can be further improved by the dimensionality reduction in specific learning tasks.


\bibliography{Ref}

\begin{thebibliography}{127}
\bibitem [{PGK()}]{PGKNote}%
  \BibitemOpen
  \href@noop {} {}\bibinfo {note} {Note that in the usual definition of good
  kernels~\cite{FourierBookStein}, there is no requirement for the positivity,
  but a $L_1$-boundness condition of the kernel is necessary, i.e., $\int_{x
  \in {\cal{X}} } | K_{\Lambda} (x) | \mathrm{d}\mu(x) \le {\cal{M}}$ as
  $\Lambda \to \infty$, with ${\cal{M}}$ a positive constant. Here for the
  density matrix learning, the positivity and normalization conditions together
  lead to the $L_1$-boundness, with ${\cal{M}} = 1$. Therefore, we do not write
  this condition explicitly.}\BibitemShut {Stop}%
\end{thebibliography}

\end{document}